\documentclass[sn-mathphys,Numbered]{sn-jnl}

\usepackage{graphicx}%
\usepackage{multirow}%
\usepackage{amsmath,amssymb,amsfonts}%
\usepackage{amsthm}%
\usepackage{mathrsfs}%
\usepackage[title]{appendix}%
\usepackage{xcolor}%
\usepackage{textcomp}%
\usepackage{manyfoot}%
\usepackage{booktabs}%
\usepackage{algorithm}%
\usepackage{algorithmicx}%
\usepackage{algpseudocode}%
\usepackage{listings}%

\theoremstyle{thmstyleone}%
\newtheorem{theorem}{Theorem}
\newtheorem{proposition}[theorem]{Proposition}%

\theoremstyle{thmstyletwo}%
\newtheorem{example}{Example}%
\newtheorem{remark}{Remark}%

\theoremstyle{thmstylethree}%
\newtheorem{definition}{Definition}%

\begin{document}

\title[Pedal and contrapedal curves of mixed-type Minkowski plane curves]{Pedal and contrapedal curves of mixed-type Minkowski plane curves}


\author[1]{\fnm{Xin} \sur{Zhao}}\email{zhaoxin@qdu.edu.cn}

\author*[2]{\fnm{Pengcheng} \sur{Li}}\email{lipc@hit.edu.cn}


\affil[1]{\orgdiv{School of Mathematics and Statistics}, \orgname{Qingdao University}, \orgaddress{
\city{Qingdao}, \postcode{266071}, \state{Shandong}, \country{P. R. China}}}

\affil*[2]{\orgdiv{Department of Mathematics}, \orgname{Harbin Institute of Technology}, \orgaddress{
\city{Weihai}, \postcode{264209}, \state{Shandong}, \country{P. R. China}}}



\abstract{Pedal and contrapedal curves are important study objects of plane curves. As for a mixed-type Minkowski plane curve, since the definitions of the pedal and contrapedal curves at lightlike points can not always be given, the investigation of them is difficult. We have done some research on the pedal curves of a mixed-type curve. In this paper, we discuss when the contrapedal curves of a mixed-type curve exist and give the definition of them when they exist. Then, we study  when the contrapedal curves of the mixed-type curve have singular points. Meanwhile, we consider the types of the points on the contrapedal curves. Moreover, we investigate the relationship between the pedal and contrapedal curves of a mixed-type curve, as well as the relationship among them and the evolute of the base curve.}

\keywords{mixed-type curve, pedal curve, contrapedal curve, evolute, lightlike point, Minkowski plane}



\maketitle

\section{Introduction}\label{sec1}
\label{intro}
Pedal and contrapedal curves are important mathematical models in classical mechanics, and they have drawn the attention of many scientists. For instance, P. Blaschke showed that the trajectory of a test particle under the influence of central and Lorentz-like forces can be translated into pedal coordinates at once without the need of solving any differential equation (see \cite{Blaschke1}). P. Blaschke, F. Blaschke and M. Blaschke investigated the orbits of a free double linkage by the technique of pedal coordinates (see \cite{Blaschke2}). They provided a geometrical construction for them and also show a surprising connection between this mechanical system and orbits around a Black Hole and solutions of Dark Kepler problem. Besides, the pedal and contrapedal curves of a curve in the Euclidean space are important and classical study objects in differential geometry. In the Euclidean plane,  the pedal curve is always defined as the locus of the foot of the perpendicular from a given point to the tangent of the base curve, and the contrapedal curve is always defined by the locus of the foot of the perpendicular from a given point to the normal of the base curve. There has been a lot of investigations about pedal curves and contrapedal curves in different spaces. Specifically, M. Bo\v{z}ek and G. Folt\'{a}n investigated the relationship between the singular points of the pedal curves of a regular curve and the inflections of the base curve in the Euclidean plane (see \cite{s-a-p}). Y. Li and D. Pei considered the pedal curves of a curve that the singular points are allowed in the sphere
(see \cite{MR3417950}). Additionally, they investigated the pedal curves of a curve that the singular points
are allowed in the Euclidean plane (see \cite{MR3778101}). Later, O. O\u{g}ulcan Tuncer, H. Ceyhan, \"{I}. G\"{o}k and F. N. Ekmekci conducted further research about the properties of the pedal curves and the
contrapedal curves of a curve in the Euclidean plane, as well as they studied the relationship between the pedal curves and the contrapedal curves (see \cite{MR3843581}).

The Minkowski space is the space-time model of Einstein special theory of relativity. It played an important
role in the development of modern physics. Therefore, the investigation of the submanifolds in the Minkowski
space and its subspaces has great significance (see \cite{MR4296313, Chenl, MR2358723, MR2735275, Oztek,
MR2056288, Tukel, MR2944524, yang, Silva}). The Minkowski space is the space with index 1. Due to the appearance of the index, the pseudo-scalar product of any non-zero vector and itself in this space may be greater than zero, less than zero or equal to zero. These three cases correspond to spacelike vectors, timelike vectors and lightlike vectors, respectively. However, for a curve in the Minkowski space, the type of a point on the curve is determined by the type of the tangent vector at this point. Among them, we call the curve composed of spacelike points or timelike points a spacelike curve or a timelike curve, and these two types of curves are collectively called non-lightlike curves.  We call a curve composed entirely of lightlike points a lightlike curve. The study of the non-lightlike curves in the Minkowski space is similar to the study of the curves in the Euclidean space. Mathematicians often select the arc-length parameter and study non-ligthlike curves by Frenet-Serret frame. There has been a wealth of research results on non-lightlike curves (see \cite{Tukel, MR2944524,MR4296313, yang}). For lightlike curves, the arc-length parameter cannot be selected because the tangent vector at each point is lightlike, the Frenet-Serret frame is not applicable anymore. In order to study lightlike curves, K. L. Duggal and D. H. Jin introduced the concept of the pseudo arc-length parameter and established Cartan frame (see \cite{MR2358723}).  By the proper research tools, there are a lot of research results about lightlike curves (see \cite{Oztek, MR4141668}).

The curves in the Minkowski space are usually not composed of a single type of points, but there are spacelike points, timelike points and lightlike points simultaneously on a curve, which is called the mixed-type curve.
As a kind of more general curves, mixed-type curves have important research value.
As the curvature at the lightlike points of a mixed-type curve cannot be defined, the traditional Frenet-Serret frame is no longer applicable. Therefore, the investigation on mixed-type curves is almost blank because of the lack of necessary tools. Until 2018, S. Izumiya, M. C. Romero Fuster and M. Takahashi gave the lightcone frame in the Minkowski plane for the first time and established fundamental theory of the regular mixed-type curves. As an application of the theory, they defined the evolute of a regular mixed curve in the Minkowski plane using the lightcone frame and explored its singular points (see \cite{MR3839951}). In 2020, T. Liu and D. Pei established the lightcone frame in the Minkowski 3-space and used the frame to study the mixed-type curves in this space (see \cite{MR4105773}).  Thus, the research progress of mixed-type curves is further promoted. As the study of the mixed-type curves in the Minkowski plane is not perfect yet, we further studied the mixed-type curves in the Minkowski plane and perfected the fundamental theory of the mixed-type type curves in this space. We investigated the $(n,m)$-cusp mixed-type curves and the evolutes of them with the modified Frenet-Serret type frame in 2021 (see \cite{MR4189889}). Later we investigated the evolutoids of the mixed-type curves with the help of the method of the region division in the Minkowski plane (see \cite{MR4357157}). In addition, we did some work on the pedal curves of the mixed-type curves in the Minkowski plane (see \cite{z-p-pedal}).
But as for the contrapedal curves of the mixed-type curves, there is no relevant investigation. Therefore,
here we devote ourselves to investigate the contrapedal curves of the mixed-type curves in the Minkowski plane. Also we investigate the relationship of the pedal and contrapedal curves of the mixed-type curves.

The brief structure of this paper is as follows. In Section \ref{sec_1}, we introduce some essential knowledge about the Minkowski plane. Then, we discuss when the contrapedal curves of the mixed-type curves exist and give the definition of the contrapedal curves when they exist in Section \ref{sec_2}. We also discuss the specific forms of the point on the contrapedal curves corresponding to the lightlike point on the base curve in this section. In addition, we consider when the contrapedal curves have singular points and investigate relationship of the types of the points on the contrapedal curves and the types of the points on the base curve. Finally, we consider the relationship between the pedal curves and contrapedal curves of a mixed-type curve, as well as the relationship among them and the evolutes of the mixed-type curve in Section \ref{sec_3}.

Unless otherwise stated, the mappings and submanifolds covered in this paper are smooth.

\section{Preliminaries}\label{sec_1}
We now review some basic conceptions about the Minkowski plane.

Let $\mathbb{R}^2=\{(u_1,u_2)\ |u_i\in\mathbb{R},\ i=1,~2\}$ be the vector space of dimension 2.
If $\mathbb{R}^2$ is endowed with the metric induced by the \emph{pseudo-scalar product}
$$\langle\boldsymbol {u},\boldsymbol {v}\rangle=-u_1v_1+u_2v_2,$$
where $\boldsymbol {u}=(u_1,u_2)$, $\boldsymbol {v}=(v_1,v_2)$,
and $\boldsymbol {u}, \boldsymbol {v}\in\mathbb{R}^2$. Then we call $(\mathbb{R}^2,\langle,\rangle)$ the
\emph{Minkowski plane} and denote it by $\mathbb{R}_1^2$. 
\par

For a non-zero vector $\boldsymbol {u}\in\mathbb{R}_1^2$, it is called \emph{spacelike},
\emph{timelike} or \emph{lightlike}, if $\langle\boldsymbol {u},\boldsymbol {u}\rangle\textgreater0$,
$\langle\boldsymbol {u},\boldsymbol {u}\rangle\textless0$ or $\langle\boldsymbol {u},\boldsymbol {u}\rangle=0$,
respectively. A  spacelike or timelike vector is called a \emph{non-lightlike} vector.

For a vector $\boldsymbol {u}\in\mathbb{R}_1^2$, if there exists a vector $\boldsymbol {v}\in\mathbb{R}_1^2$,
such that $\langle\boldsymbol {u},\boldsymbol {v}\rangle=0$, we say $\boldsymbol {v}$ is
\emph{pseudo-orthogonal} to $\boldsymbol {u}$.

The \emph{norm} of a vector $\boldsymbol {u}=(u_1,u_2)\in\mathbb{R}_1^2$ is defined by
$$\|\boldsymbol {u}\|=\sqrt{|\langle\boldsymbol {u},\boldsymbol {u}\rangle|}.$$
The pseudo-orthogonal vector of $\boldsymbol {u}$ is given by $\boldsymbol {u}^\perp=(u_2,u_1)$.
By definition, $\boldsymbol {u}$ and $\boldsymbol {u}^\perp$ are pseudo-orthogonal to each other,
and $$\|\boldsymbol {u}\|=\|\boldsymbol {u}^\perp\|.$$ 

It's obvious that $\boldsymbol {u}^\perp=\pm\boldsymbol {u}$ if and only if $\boldsymbol {u}$ is lightlike, and $\boldsymbol {u}^\perp$ is timelike (resp. spacelike) if and only if $\boldsymbol {u}$ is spacelike (resp. timelike). 
\par

Denote $\mathbb{L}^+=(1,1)$ and $\mathbb{L}^-=(1,-1)$. Then $\mathbb{L}^+$ and $\mathbb{L}^-$ are independent lightlike vectors, and $\langle\mathbb{L}^+,\mathbb{L}^-\rangle=-2$. We call $\{\mathbb{L}^+,\mathbb{L}^-\}$ a \emph{lightcone frame} in $\mathbb{R}_1^2$. It is given by S. Izumiya, M. C. Romero Fuster and M. Takahashi in \cite{MR3839951}. \par 

Let $\gamma:I\rightarrow\mathbb{R}_1^2$ be a regular mixed-type curve. There exists a smooth map $(\alpha,\beta):I\rightarrow\mathbb{R}^2\backslash \{\boldsymbol{0}\}$ such that
\begin{align} \nonumber
\dot\gamma(t)=\alpha(t){\mathbb{L}^+}+\beta(t){\mathbb{L}^-},
\end{align}
for all $t\in{I}$. We say that \emph{a regular curve $\gamma$ with the lightlike tangential data $(\alpha,\beta)$} if the above equation holds. Then we have $$\dot\gamma(t)^\bot=\alpha(t){\mathbb{L}^+}-\beta(t){\mathbb{L}^-}.$$
Since $$\langle\dot\gamma(t),\dot\gamma(t)\rangle=-4\alpha(t)\beta(t),$$ $\gamma(t_0)$ is a spacelike (resp. lightlike or timelike) point if and only if $\alpha(t_0)\beta(t_0)<0$ (resp. $=0$ or $>0$).

\begin{definition}
Let $\gamma:I\rightarrow\mathbb{R}_1^2$ be a regular mixed-type curve. We call a point $\gamma(t_0)$ an inflection if $\langle\ddot\gamma(t_0),\dot\gamma(t_0)^\perp\rangle=0$.
\end{definition}
On the basis of the above definition, if $\langle\ddot\gamma(t_0),\dot\gamma(t_0)^\perp\rangle'\neq0$, we call $\gamma(t_0)$ an ordinary inflection.

If we choose the lightcone frame $\{\mathbb{L}^+,\mathbb{L}^-\}$ and the lightlike tangential data $(\alpha,\beta)$, then we can obtain that $\gamma(t_0)$ is an inflection point if and only if $$\dot\alpha(t_0)\beta(t_0)-\alpha(t_0)\dot\beta(t_0)=0,$$ and $\gamma(t_0)$ is an ordinary inflection point means not only $$\dot\alpha(t_0)\beta(t_0)-\alpha(t_0)\dot\beta(t_0)=0,$$ but also $$\ddot\alpha(t_0)\beta(t_0)-\alpha(t_0)\ddot\beta(t_0)\neq0.$$

Let $\gamma:I\rightarrow\mathbb{R}_1^2$ be a regular mixed-type curve without inflections. In \cite{MR3839951}, we have known that the evolute $Ev:I\rightarrow\mathbb{R}_1^2$ of $\gamma$ with the lightlike data $(\alpha,\beta)$ is defined as
$$Ev(\gamma)(t)=\gamma(t)-\frac{2\alpha(t)\beta(t)}{\dot\alpha(t)\beta(t)-\alpha(t)\dot\beta(t)}(\alpha(t)\mathbb{L}^+-\beta(t)\mathbb{L}^-).$$

\section{Contrapedal curves of the mixed-type curves}\label{sec_2}
For a regular curve in the Minkowski plane, the pedal curve of it is always defined by the pseudo-orthogonal projection of a fixed point on the tangent lines of the base curve, and the contrapedal curve of it is always defined by the pseudo-orthogonal projection of a fixed point on the normal lines of the base curve. We have investigated the pedal curves of the mixed-type curves in \cite{z-p-pedal}. Herein, we consider the contrapedal curves of the mixed-type curves.

Similar to the definition of the pedal curve of the regular mixed-type curve, we can define the contrapedal curve of a regular mixed-type curve as follows.

\begin{definition}
Let $\gamma:I\rightarrow\mathbb{R}_1^2$ be a regular mixed-type curve and $\boldsymbol{Q}$ be a point in $\mathbb{R}_1^2$. If $\boldsymbol{Q}$ is consistent with the lightlike point or $\boldsymbol{Q}$ is on the tangent line of the lightlike point, then the contrapedal curve $CPe(\gamma)(t)$ of the base curve $\gamma(t)$ is given by
\begin{align}\label{eq1}
CPe(\gamma)(t)=\gamma(t)+\frac{\langle\boldsymbol{Q}-\gamma(t),\alpha(t){\mathbb{L}^+}-\beta(t){\mathbb{L}^-}\rangle}
{4\alpha(t)\beta(t)}(\alpha(t){\mathbb{L}^+}-\beta(t){\mathbb{L}^-}).
\end{align}
\end{definition}

If $\gamma(t_0)$ is a non-lightlike point, then $CPe(\gamma)(t_0)$ satisfies above form obviously.

If $\gamma(t_0)$ is a lightlike point, then $\alpha(t_0)\beta(t_0)=0$, and we suppose that $\boldsymbol{Q}$ is consistent with the lightlike point or $\boldsymbol{Q}$ is on the tangent line of the lightlike point. In this case we define $Pe(\gamma)(t_0)$ as $\displaystyle\lim_{t\to t_0}Pe(\gamma)(t)$, and the specific forms of the contrapedal curve at $\gamma(t_0)$ are as follows.

\textbf{CASE I}. Suppose that $\alpha(t_0)\neq0$ and $\beta(t_0)=0$, then $\boldsymbol{Q}$ is consistent with $\gamma(t_0)$ or $\boldsymbol{Q}$ is on the tangent line of $\gamma(t_0)$ is exactly $\langle\boldsymbol{Q}-\gamma(t),{\mathbb{L}^+}\rangle=0$.

If $\boldsymbol{Q}$ is consistent with $\gamma(t_0)$, then
$$CPe(\gamma)(t_0)=\gamma(t_0).$$

If $\boldsymbol{Q}$ is on the tangent line of the lightlike point, then
$$CPe(\gamma)(t_0)=\gamma(t_0)-\frac{1}{4}\langle\boldsymbol{Q}-\gamma(t_0),{\mathbb{L}^-}\rangle{\mathbb{L}^+}.$$

\textbf{CASE II}. Suppose that $\alpha(t_0)=0$ and $\beta(t_0)\neq0$, then $\boldsymbol{Q}$ is consistent with $\gamma(t_0)$ or $\boldsymbol{Q}$ is on the tangent line of $\gamma(t_0)$ is exactly $\langle\boldsymbol{Q}-\gamma(t),{\mathbb{L}^-}\rangle=0$.

If $\boldsymbol{Q}$ is consistent with $\gamma(t_0)$, then
$$CPe(\gamma)(t_0)=\gamma(t_0).$$

If $\boldsymbol{Q}$ is on the tangent line of the lightlike point, then
$$CPe(\gamma)(t_0)=\gamma(t_0)-\frac{1}{4}\langle\boldsymbol{Q}-\gamma(t_0),{\mathbb{L}^+}\rangle{\mathbb{L}^-}.$$

\begin{remark}
Let $\gamma:I\rightarrow\mathbb{R}_1^2$ be a regular mixed-type curve, $\boldsymbol{Q}$ be a point in $\mathbb{R}_1^2$ and $CPe(\gamma):I\rightarrow\mathbb{R}_1^2$ be the contrapedal curve of $\gamma$. Suppose that $\gamma(t_0)$ is a lightlike point, if $\boldsymbol{Q}$ is neither consistent with $\gamma(t_0)$ nor on the tangent line of $\gamma(t_0)$. Since $\alpha(t_0)=0$ or $\beta(t_0)=0$ holds, $\displaystyle\lim_{t\to t_0}\dfrac{\langle\boldsymbol{Q}-\gamma(t),\alpha(t){\mathbb{L}^+}+\beta(t){\mathbb{L}^-}\rangle}
{4\alpha(t)\beta(t)}=\infty$. Therefore, $CPe(\gamma)(t_0)$ is asymptotic with the lightlike line of $\mathbb{L}^+$ or $\mathbb{L}^-$. Specifically, when $\langle\boldsymbol{Q}-\gamma(t),{\mathbb{L}^+}\rangle\neq0$,
$CPe(\gamma)(t_0)$ is asymptotic with the lightlike line along the positive or negative direction of $\mathbb{L}^+.$ When $\langle\boldsymbol{Q}-\gamma(t),{\mathbb{L}^-}\rangle\neq0$, $CPe(\gamma)(t_0)$ is asymptotic with the lightlike line along the positive or negative direction of $\mathbb{L}^-.$
\end{remark}

Next, we consider when the contrapedal curve of a regular mixed-type curve have singular points and we have the following theorem.

\begin{theorem}
Let $\gamma:I\rightarrow\mathbb{R}_1^2$ be a regular mixed-type curve, $\boldsymbol{Q}$ be a point in $\mathbb{R}_1^2$ and $CPe(\gamma):I\rightarrow\mathbb{R}_1^2$ be the contrapedal curve of $\gamma$. Then,

$(1)$ if $\gamma(t_0)$ is a non-lightlike point, then $CPe(\gamma)(t_0)$ is a singular point if and only if
$$\alpha(t_0)+\dfrac{\dot\alpha(t_0)\beta(t_0)-\alpha(t_0)\dot\beta(t_0)}{4\beta^2(t_0)}\langle\boldsymbol{Q}-\gamma(t_0),\mathbb{L}^+\rangle=0$$ and $$\beta(t_0)-\dfrac{\dot\alpha(t_0)\beta(t_0)-\alpha(t_0)\dot\beta(t_0)}{4\alpha^2(t_0)}\langle\boldsymbol{Q}-\gamma(t_0),\mathbb{L}^-\rangle=0;$$

$(2)$ if $\gamma(t_0)$ is a lightlike point, and $\boldsymbol{Q}$ is consistent with $\gamma(t_0)$ or $\boldsymbol{Q}$ is on the tangent line of $\gamma(t_0)$, then $CPe(\gamma)(t_0)$ is regular.
\end{theorem}

\begin{proof}
Since the contrapedal curve of the mixed-type curve $\gamma(t)$ is given by the formula (\ref{eq1}), by calculations, we can obtain
\begin{align}\label{eq2}\nonumber
&\dot CPe(\gamma)(t)\\ \nonumber
=&(\alpha(t)+\frac{\dot\alpha(t)\beta(t)-\alpha(t)\dot\beta(t)}{4\beta^2(t)}\langle\boldsymbol{Q}-\gamma(t),\mathbb{L}^+\rangle)\mathbb{L}^+\\
+&(\beta(t)-\frac{\dot\alpha(t)\beta(t)-\alpha(t)\dot\beta(t)}{4\alpha^2(t)}\langle\boldsymbol{Q}-\gamma(t),\mathbb{L}^-\rangle)\mathbb{L}^-.
\end{align}
When $\gamma(t_0)$ is a non-lightlike point, $\dot CPe(\gamma)(t_0)=0$ if and only if
$$\alpha(t_0)+\frac{\dot\alpha(t_0)\beta(t_0)-\alpha(t_0)\dot\beta(t_0)}{4\beta^2(t_0)}\langle\boldsymbol{Q}-\gamma(t_0),\mathbb{L}^+\rangle=0$$
and
$$\beta(t_0)-\frac{\dot\alpha(t_0)\beta(t_0)-\alpha(t_0)\dot\beta(t_0)}{4\alpha^2(t_0)}\langle\boldsymbol{Q}-\gamma(t_0),\mathbb{L}^-\rangle=0.$$
When $\gamma(t_0)$ is a lightlike point, we consider the condition that $\alpha(t_0)\neq0$, $\beta(t_0)=0$ first. Since $\beta(t_0)=0$, we cannot calculate $\dfrac{\dot\alpha(t_0)\beta(t_0)-\alpha(t_0)\dot\beta(t_0)}{4\beta^2(t_0)}\langle\boldsymbol{Q}-\gamma(t_0),\mathbb{L}^+\rangle\mathbb{L}^+$. We have known that when $\langle\boldsymbol{Q}-\gamma(t),{\mathbb{L}^+}\rangle\neq0$,
$CPe(\gamma)(t_0)$ is asymptotic with the lightlike line along the positive or negative direction of $\mathbb{L}^+.$ Thus we consider the condition that $\langle\boldsymbol{Q}-\gamma(t_0),\mathbb{L}^+\rangle=0$.

We define $\dfrac{(\dot\alpha(t_0)\beta(t_0)-\alpha(t_0)\dot\beta(t_0))\langle\boldsymbol{Q}-\gamma(t_0),\mathbb{L}^+\rangle}{\beta^2(t_0)}$ as $$\displaystyle\lim_{t\to t_0}\frac{(\dot\alpha(t)\beta(t)-\alpha(t)\dot\beta(t))\langle\boldsymbol{Q}-\gamma(t),\mathbb{L}^+\rangle}{\beta^2(t)}.$$ By calculation, we can get that
$$\displaystyle\lim_{t\to t_0}\frac{(\dot\alpha(t)\beta(t)-\alpha(t)\dot\beta(t))\langle\boldsymbol{Q}-\gamma(t),\mathbb{L}^+\rangle}{\beta^2(t)}
\neq0.$$
Therefore, $CPe(\gamma)(t_0)$ is a regular point.

For the condition that $\alpha(t_0)=0$, $\beta(t_0)\neq0$ and $\langle\boldsymbol{Q}-\gamma(t),{\mathbb{L}^-}\rangle=0$, we can get $CPe(\gamma)(t_0)$ is a regular point similarly.
\end{proof}

For the sake of the following discussion, let
\begin{align}\nonumber
&\Omega(t)\\ \nonumber
=&1+\dfrac{\dot\alpha(t)\beta(t)-\alpha(t)\dot\beta(t)}{4\alpha(t)\beta^2(t)}\langle\boldsymbol{Q}-\gamma(t),\mathbb{L}^+\rangle
-\dfrac{\dot\alpha(t)\beta(t)-\alpha(t)\dot\beta(t)}{4\alpha^2(t)\beta(t)}\langle\boldsymbol{Q}-\gamma(t),\mathbb{L}^-\rangle\\
&-\dfrac{(\dot\alpha(t)\beta(t)-\alpha(t)\dot\beta(t))^2}{16\alpha^3(t)\beta^3(t)}
\langle\boldsymbol{Q}-\gamma(t_0),\mathbb{L}^+\rangle\langle\boldsymbol{Q}-\gamma(t),\mathbb{L}^-\rangle, \nonumber
\end{align}
then the following proposition shows the type of points of the contrapedal curve of a mixed-type curve.

\begin{proposition}
Let $\gamma:I\rightarrow\mathbb{R}_1^2$ be a regular mixed-type curve, $\boldsymbol{Q}$ be a point in $\mathbb{R}_1^2$ and $CPe(\gamma):I\rightarrow\mathbb{R}_1^2$ be the contrapedal curve of $\gamma$. For a regular point  $CPe(\gamma)(t_0)$, we have  the following consequences.\\
$(1)$ When $\gamma(t_0)$ is a non-lightlike point,

${\rm(\romannumeral1)}$ if $\Omega(t_0)>0$, then $CPe(\gamma)(t_0)$ is a spacelike (or timelike) point if and

\ \ \ only if $\gamma(t_0)$ is a spacelike (or timelike) point;

${\rm(\romannumeral2)}$ if $\Omega(t_0)<0$, then $CPe(\gamma)(t_0)$ is a spacelike (or timelike) point if and

\ \ \ only if $\gamma(t_0)$ is a timelike (or spacelike) point;

${\rm(\romannumeral3)}$ if $\Omega(t_0)=0$, then $CPe(\gamma)(t_0)$ is a lightlike point.\\
$(2)$ When $\gamma(t_0)$ is a lightlike point, $\alpha(t_0)\neq0$, $\beta(t_0)=0$ and $\langle\boldsymbol{Q}-\gamma(t_0),\mathbb{L}^+\rangle=0$,

  ${\rm(\romannumeral1)}$ suppose that $\gamma(t_0)$ is not the inflection of $\gamma$,

  \ \ \ $(a)$ $CPe(\gamma)(t_0)$ is a lightlike point if and only if $\boldsymbol{Q}$ is consistent with

  \ \ \ \ \ \ \ $\gamma(t_0)$;

  \ \ \ $(b)$ $CPe(\gamma)(t_0)$ is a non-lightlike point if and only if $\boldsymbol{Q}$ is on the tangent

  \ \ \ \ \ \ \ line of $\gamma(t_0)$. Moreover, $CPe(\gamma)(t_0)$ is a spacelike (or, timelike) point

  \ \ \ \ \ \ \ if and only if $\langle\boldsymbol{Q}-\gamma(t_0),\mathbb{L}^-\rangle\dot\beta(t_0)<0$ (or, $\langle\boldsymbol{Q}
  -\gamma(t_0),\mathbb{L}^-\rangle\dot\beta(t_0)>0$);

 ${\rm(\romannumeral2)}$ suppose that $\gamma(t_0)$ is an ordinary inflection of $\gamma$, then $CPe(\gamma)(t_0)$ is

  \ \ \ always lightlike.\\
$(3)$ When $\gamma(t_0)$ is a lightlike point, $\alpha(t_0)=0$, $\beta(t_0)\neq0$ and $\langle\boldsymbol{Q}-\gamma(t_0),\mathbb{L}^-\rangle=0$,

  ${\rm(\romannumeral1)}$ suppose that $\gamma(t_0)$ is not the inflection of $\gamma$,

  \ \ \ $(a)$ $CPe(\gamma)(t_0)$ is a lightlike point if and only if $\boldsymbol{Q}$ is consistent with

  \ \ \ \ \ \ \ $\gamma(t_0)$;

  \ \ \ $(b)$ $CPe(\gamma)(t_0)$ is a non-lightlike point if and only if $\boldsymbol{Q}$ is on the tangent

  \ \ \ \ \ \ \ line of $\gamma(t_0)$. Moreover, $CPe(\gamma)(t_0)$ is a spacelike (or, timelike) point

  \ \ \ \ \ \ \ if and only if $\langle\boldsymbol{Q}-\gamma(t_0),\mathbb{L}^+\rangle\dot\alpha(t_0)<0$ (or, $\langle\boldsymbol{Q}
  -\gamma(t_0),\mathbb{L}^+\rangle\dot\alpha(t_0)>0$);

 ${\rm(\romannumeral2)}$ suppose that $\gamma(t_0)$ is an ordinary inflection of $\gamma$, then $CPe(\gamma)(t_0)$ is

  \ \ \ always lightlike.
\end{proposition}
\begin{proof}
As $\dot CPe(\gamma)(t)$ is given by formula (\ref{eq2}), by calculations we can get
\begin{align}\nonumber
&\langle \dot CPe(\gamma)(t),\dot CPe(\gamma)(t)\rangle\\ \nonumber
=&-4\alpha(t)\beta(t)(1+\frac{\dot\alpha(t)\beta(t)-\alpha(t)\dot\beta(t)}{4\alpha(t)\beta^2(t)}
\langle\boldsymbol{Q}-\gamma(t),\mathbb{L}^+\rangle\\ \nonumber
&-\frac{\dot\alpha(t)\beta(t)-\alpha(t)\dot\beta(t)}{4\alpha^2(t)\beta(t)}
\langle\boldsymbol{Q}-\gamma(t),\mathbb{L}^-\rangle\\ \nonumber
&-\frac{(\dot\alpha(t)\beta(t)-\alpha(t)\dot\beta(t))^2}{16\alpha^3(t)\beta^3(t)}
\langle\boldsymbol{Q}-\gamma(t),\mathbb{L}^+\rangle\langle\boldsymbol{Q}-\gamma(t),\mathbb{L}^-\rangle)\\ \nonumber
=&-4\alpha(t)\beta(t)\Omega(t).
\end{align}
Then we can obtain the type of $CPe(\gamma)(t_0)$ easily.
\end{proof}

Following, we give three examples to present the features of the contrapedal curve of the regular mixed-type curve, especially at the lightlike point of the base curve.

\begin{example}
Let $\gamma:[0,2\pi)\rightarrow\mathbb{R}_1^2$ be a regular mixed-type curve, where
$$\gamma(t)=\left(2\cos t, \frac{2\sqrt{3}}{3}\sin t\right).$$
When $t_0=\dfrac{5}{6}\pi$, $\gamma(t_0)$ is a lightlike point. See the blue curve in Figure \ref{f_1}.

If $\boldsymbol{Q}=(0,0)$, then $\langle\boldsymbol{Q}-\gamma(t_0),\mathbb{L}^+\rangle\neq0$, the contrapedal curve of $\gamma(t)$
is $$CPe(\gamma)(t)=\left(-\frac{8\sin^2t\cos t}{\cos^2t-3\sin^2t},\frac{8\sqrt{3}\sin t\cos^2t}{3(\cos^2t-3\sin^2t)}\right).$$
In this case, $Pe(\gamma)(t_0)$ is asymptotic with the lightlike line along the positive and negative direction of $\mathbb{L}^+$. See the green curve in Figure \ref{f_1}.

If $\boldsymbol{Q}=\left(-\sqrt{3},\dfrac{\sqrt{3}}{3}\right)$, then $\boldsymbol{Q}$ coincides with $\gamma(t_0)$, the contrapedal curve of $\gamma(t)$
is
\begin{align}\nonumber
&CPe(\gamma)(t)\\ \nonumber
=&\left(\frac{-\sqrt{3}\cos^2t+\sin t\cos t-8\sin^2t\cos t}{\cos^2t-3\sin^2t},\frac{-3\sqrt{3}\sin^2t+9\sin t\cos t+8\sqrt{3}\sin t\cos^2t}{3(\cos^2t-3\sin^2t)}\right).
\end{align}

In this case, $CPe(\gamma)(t_0)$ is a lightlike point. See the orange dashed curve in Figure \ref{f_1}.

If $\boldsymbol{Q}=\left(0,\dfrac{4\sqrt{3}}{3}\right)$, then $\boldsymbol{Q}$ is on the tangent line of $\gamma(t_0)$, the pedal curve of $\gamma(t)$
is $$CPe(\gamma)(t)=\left(\frac{-8\sin^2t\cos t+4\sin t\cos t}{\cos^2t-3\sin^2t},\frac{-12\sqrt{3}\sin^2t+8\sqrt{3}\sin t\cos^2t}{3(\cos^2t-3\sin^2t)}\right).$$
In this case, $CPe(\gamma)(t_0)$ is a spacelike point. See the red dashed curve in Figure \ref{f_1}.
\begin{figure}[ h ]
\centering
\includegraphics [width =6cm]{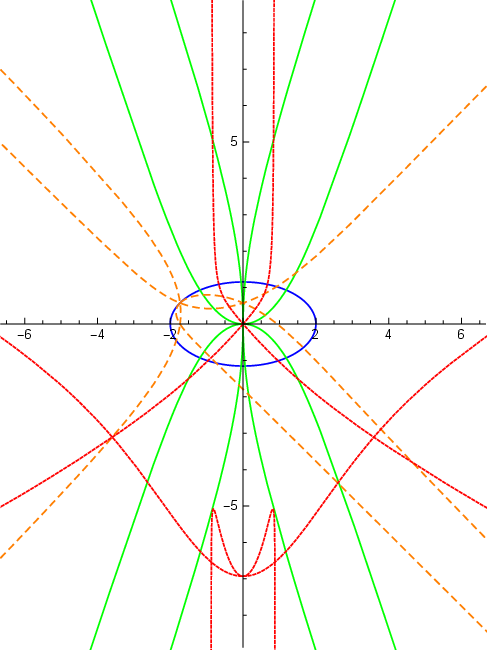}
\caption{The mixed-type curve (blue) and its contrapedal curves.}\label{f_1}
\end{figure}
\end{example}
\pagebreak
\begin{example}
Let $\gamma:(-1,1)\rightarrow\mathbb{R}_1^2$ be a regular mixed-type curve, where
$$\gamma(t)=(t, t^2).$$
When $t_0=\dfrac{1}{2}\pi$, $\gamma(t_0)$ is a lightlike point. See the blue curve in Figure \ref{f_2}.
If $\boldsymbol{Q}=(0,0)$, then $\langle\boldsymbol{Q}-\gamma(t_0),\mathbb{L}^+\rangle\neq0$, the contrapedal curve of $\gamma(t)$
is $$CPe(\gamma)(t)=\left(\frac{2t^3-t}{4t^2-1},\frac{4t^4-2t^2}{4t^2-1}\right).$$
In this case, $CPe(\gamma)(t_0)$ is asymptotic with the lightlike line along the positive and negative direction of $\mathbb{L}^+$. See the green curve in Figure \ref{f_2}.

If $\boldsymbol{Q}=\left(\dfrac{1}{2},\dfrac{1}{4}\right)$, then $\boldsymbol{Q}$ coincides with $\gamma(t_0)$, the contrapedal curve of $\gamma(t)$
is $$CPe(\gamma)(t)=\left(\frac{2t^2+3t}{4t+2},\frac{8t^3+4t^2-2t+1}{8t+4}\right).$$
In this case, $Pe(\gamma)(t_0)$ is a lightlike point. See the orange dashed curve in Figure \ref{f_2}.

If $\boldsymbol{Q}=\left(1,\dfrac{3}{4}\right)$, then $\boldsymbol{Q}$ is on the tangent line of $\gamma(t_0)$, the pedal curve of $\gamma(t)$
is $$CPe(\gamma)(t)=\left(\frac{2t^2+5t}{4t+2},\frac{8t^3+4t^2-2t+3}{8t+4}\right).$$
In this case, $CPe(\gamma)(t_0)$ is a timelike point. See the red dashed curve in Figure \ref{f_2}.
\begin{figure}[ h ]
\centering
\includegraphics [width =6cm]{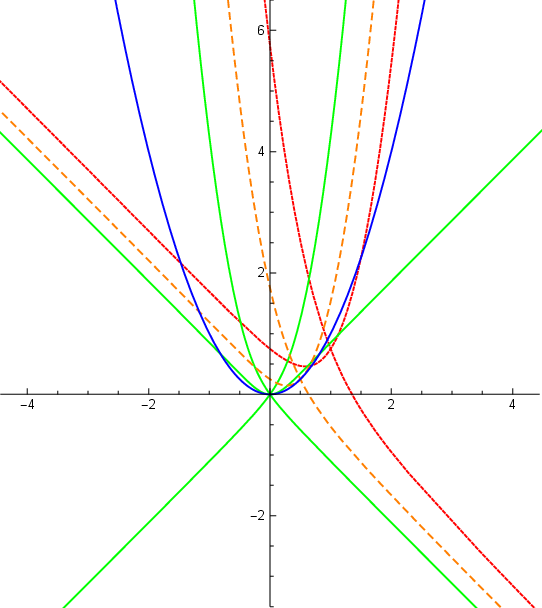}
\caption{The mixed-type curve (blue) and its contrapedal curves.}\label{f_2}
\end{figure}
\end{example}
\pagebreak
\begin{example}
Let $\gamma:(-1,1)\rightarrow\mathbb{R}_1^2$ be a regular mixed-type curve, where
$$\gamma(t)=(t, t^3+t).$$
When $t_0=0$, $\gamma(t_0)$ is a lightlike point and it is an ordinary inflection. See the blue curve in Figure \ref{f_3}.

\begin{figure}[h]
\centering
\includegraphics [width =3.7cm]{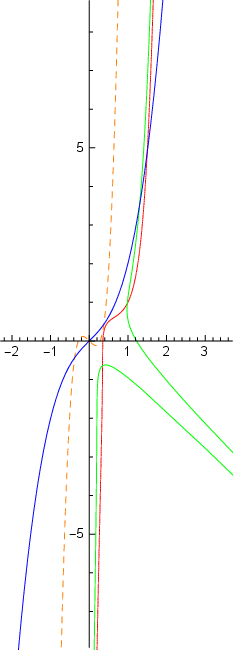}
\caption{The mixed-type curve (blue) and its contrapedal curves.}\label{f_3}
\end{figure}

If $\boldsymbol{Q}=(1,2)$, then $\langle\boldsymbol{Q}-\gamma(t_0),\mathbb{L}^+\rangle\neq0$, the contrapedal curve of $\gamma(t)$
is $$CPe(\gamma)(t)=\left(\frac{3t^5+9t^4+4t^3+1}{9t^4+6t^2},\frac{9t^7+15t^5+4t^3+3t^2-1}{9t^4+6t^2}\right).$$
In this case, $CPe(\gamma)(t_0)$ is asymptotic with the lightlike line along the positive and negative direction of $\mathbb{L}^+$. See the green curve in Figure \ref{f_3}.

If $\boldsymbol{Q}=(0,0)$, then $\boldsymbol{Q}$ coincides with $\gamma(t_0)$, the contrapedal curve of $\gamma(t)$
is $$CPe(\gamma)(t)=\left(\frac{3t^3+4t}{9t^2+6},\frac{9t^5+15t^3-4t}{9t^2+6}\right).$$
In this case, $CPe(\gamma)(t_0)$ is a lightlike point. See the orange dashed curve in Figure \ref{f_3}.

If $\boldsymbol{Q}=(1,1)$, then $\boldsymbol{Q}$ is on the tangent line of $\gamma(t_0)$, the contrapedal curve of $\gamma(t)$
is $$CPe(\gamma)(t)=\left(\frac{3t^3+9t^2+4t+3}{9t^2+6},\frac{9t^5+15t^3+4t+3}{9t^2+6}\right).$$
In this case, $CPe(\gamma)(t_0)$ is a lightlike point. See the red dashed curve in Figure \ref{f_3}.
\end{example}

\section{Relationships among pedal curves, contrapedal curves and evolutes of mixed-type curves}\label{sec_3}
As the pedal curves, contrapedal curves and evolutes are all significant curves closely related to the base curve, we would like to investigate the relationships among them in this section. Firstly we study the relationship between the pedal curves and contrapedal curves of a mixed-type curve.

We consider $Pe(\gamma)(t_0)$ and $CPe(\gamma)(t_0)$ when $\gamma(t_0)$ is a lightlike point, and we have the following proposition.

\begin{proposition}\label{prop1}
Let $\gamma:I\rightarrow\mathbb{R}_1^2$ be a regular mixed-type curve, $\boldsymbol{Q}$ be a point in $\mathbb{R}_1^2$, $Pe({\gamma}):I\rightarrow\mathbb{R}_1^2$ be the pedal curve of $\gamma$ and $CPe({\gamma}):I\rightarrow\mathbb{R}_1^2$ be the contrapedal curve of $\gamma$. Suppose that $\gamma(t_0)$ is a lightlike point, if $\boldsymbol{Q}$ is on the tangent line of $\gamma(t_0)$, then $Pe(\gamma)(t_0)=CPe(\gamma)(t_0)$.
\end{proposition}

\begin{proof}
Suppose that $\gamma(t_0)$ is a non-lightlike point. When $\alpha(t_0)\neq0$ and $\beta(t_0)=0$, according to the definitions of the pedal curves and contrapedal curves of a regular mixed-type curve, if $\boldsymbol{Q}$ is on the tangent line of $\gamma(t_0)$, then
$$Pe(\gamma)(t_0)=CPe(\gamma)(t_0)
=\gamma(t_0)-\frac{1}{4}\langle\boldsymbol{Q}-\gamma(t_0),{\mathbb{L}^-}\rangle{\mathbb{L}^+}.$$

Similarly, when $\alpha(t_0)=0$ and $\beta(t_0)\neq0$, if $\boldsymbol{Q}$ is on the tangent line of $\gamma(t_0)$, then
$$Pe(\gamma)(t_0)=CPe(\gamma)(t_0)
=\gamma(t_0)-\frac{1}{4}\langle\boldsymbol{Q}-\gamma(t_0),{\mathbb{L}^+}\rangle{\mathbb{L}^-}.$$
Therefore, if $\gamma(t_0)$ is a lightlike point and $\boldsymbol{Q}$ is on the tangent line of $\gamma(t_0)$, then $Pe(\gamma)(t_0)=CPe(\gamma)(t_0)$.
\end{proof}


The following example can show this property well.

\begin{example}
Let $\gamma:(-1,1)\rightarrow\mathbb{R}_1^2$ be a regular mixed-type curve, where
$$\gamma(t)=(t, t^3+t).$$
When $t_0=0$, $\gamma(t_0)$ is a lightlike point. See the blue curve in Figure \ref{f_4}.

Let $\boldsymbol{Q}=(1,1)$, then $\boldsymbol{Q}$ is on the tangent line of $\gamma(t_0)$. In this case, the pedal curve of $\gamma(t)$
is $$Pe(\gamma)(t)=\left(\frac{6t^3+2t+3}{9t^2+6},\frac{9t^2+2t+3}{9t^2+6}\right).$$
See the red dashed curve in Figure \ref{f_4}.
The contrapedal curve of $\gamma(t)$
is $$CPe(\gamma)(t)=\left(\frac{3t^3+9t^2+4t+3}{9t^2+6},\frac{9t^5+15t^3+4t+3}{9t^2+6}\right).$$
See the orange dashed curve in Figure \ref{f_4}.

By calculations we can obtain that
$$Pe(\gamma)(t_0)=CPe(\gamma)(t_0)=\left(-\frac{1}{2},\frac{1}{2}\right).$$
See the point $\boldsymbol{P}$ in Figure \ref{f_4}.
\begin{figure}[ h ]
\centering
\includegraphics [width =7.4cm]{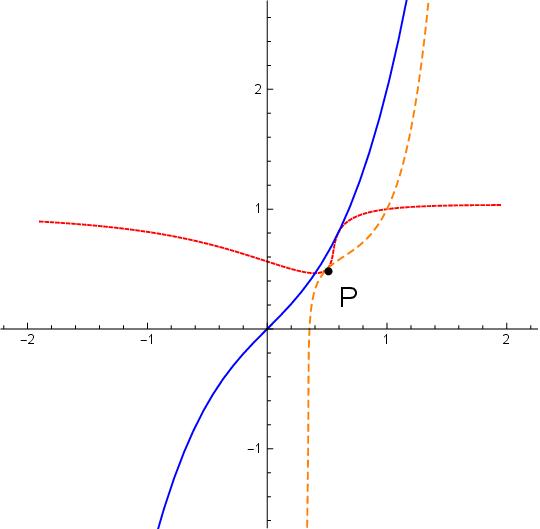}
\caption{The mixed-type curve (blue), its pedal curve and contrapedal curve.}\label{f_4}
\end{figure}
\end{example}

\begin{proposition}
Let $\gamma:I\rightarrow\mathbb{R}_1^2$ be a regular mixed-type curve, $\boldsymbol{Q}$ be a point in $\mathbb{R}_1^2$, $Pe(\gamma):I\rightarrow\mathbb{R}_1^2$ be the pedal curve of $\gamma$ and $CPe(\gamma):I\rightarrow\mathbb{R}_1^2$ be the contrapedal curve of $\gamma$. Suppose that $\gamma(t_0)$ is a non-lightlike point.

$(1)$ If $\boldsymbol{Q}$ is on the tangent line of $\gamma(t_0)$, then $Pe(\gamma)(t_0)$ is consistent with $\boldsymbol{Q}$.

$(2)$ If $\boldsymbol{Q}$ is on the normal line of $\gamma(t_0)$, then $CPe(\gamma)(t_0)$ is consistent with $\boldsymbol{Q}$.
\end{proposition}

\begin{proof}
Firstly, suppose that $\boldsymbol{Q}$ is on the tangent line of a non-lightlike point $\gamma(t_0)$, then we have $\boldsymbol{Q}-\gamma(t_0)$ and $\alpha(t_0){\mathbb{L}^+}+\beta(t_0){\mathbb{L}^-}$ are linearly dependent. Therefore, there exists $\lambda \in \mathbb{R}$, such that
$$\boldsymbol{Q}-\gamma(t_0)=\lambda(\alpha(t_0){\mathbb{L}^+}+\beta(t_0)\mathbb{L}^-).$$
Then, we have
\begin{align}\nonumber
&Pe(\gamma)(t_0)\\ \nonumber
=&\gamma(t_0)-\lambda\frac{\langle\alpha(t_0){\mathbb{L}^+}+\beta(t_0){\mathbb{L}^-},
\alpha(t_0){\mathbb{L}^+}+\beta(t_0){\mathbb{L}^-}\rangle}{4\alpha(t_0)\beta(t_0)}(\alpha(t_0){\mathbb{L}^+}+\beta(t_0){\mathbb{L}^-})\\ \nonumber
=&\gamma(t_0)-\lambda\frac{-4\alpha(t_0)\beta(t_0)}{4\alpha(t_0)\beta(t_0)}(\alpha(t_0){\mathbb{L}^+}+\beta(t_0){\mathbb{L}^-})\\ \nonumber
=&\gamma(t_0)+\lambda(\alpha(t_0){\mathbb{L}^+}+\beta(t_0){\mathbb{L}^-})\\ \nonumber
=&\boldsymbol{Q}.
\end{align}
Thus, $Pe(\gamma)(t_0)$ is consistent with $\boldsymbol{Q}$.

Then, we consider that $\boldsymbol{Q}$ is on the normal line of $\gamma(t_0)$. In this case, $\boldsymbol{Q}-\gamma(t_0)$ and $\alpha(t_0){\mathbb{L}^+}-\beta(t_0){\mathbb{L}^-}$ are linearly dependent, and there exists $\eta \in \mathbb{R}$, such that
$$\boldsymbol{Q}-\gamma(t_0)=\eta(\alpha(t_0){\mathbb{L}^+}-\beta(t_0)\mathbb{L}^-).$$
Therefore,
\begin{align}\nonumber
&CPe(\gamma)(t_0)\\ \nonumber
=&\gamma(t)+\eta\frac{\langle\alpha(t_0){\mathbb{L}^+}-\beta(t_0){\mathbb{L}^-},\alpha(t_0){\mathbb{L}^+}-\beta(t_0){\mathbb{L}^-}\rangle}
{4\alpha(t_0)\beta(t_0)}(\alpha(t_0){\mathbb{L}^+}-\beta(t_0){\mathbb{L}^-})\\ \nonumber
=&\gamma(t_0)-\eta\frac{-4\alpha(t_0)\beta(t_0)}{4\alpha(t_0)\beta(t_0)}(\alpha(t_0){\mathbb{L}^+}-\beta(t_0){\mathbb{L}^-})\\ \nonumber
=&\gamma(t_0)+\eta(\alpha(t_0){\mathbb{L}^+}-\beta(t_0){\mathbb{L}^-})\\ \nonumber
=&\boldsymbol{Q}.
\end{align}
Hence, $CPe(\gamma)(t_0)$ is consistent with $\boldsymbol{Q}$.
\end{proof}

We give an example to explain above proposition.

\begin{example}
Let $\gamma:(-1,1)\rightarrow\mathbb{R}_1^2$ be a regular mixed-type curve, where
$$\gamma(t)=(t, t^2).$$
When $t_0=1$, $\gamma(t_0)$ is a spacelike point. See the blue curve in Figure \ref{f_5}.

Let $\boldsymbol{Q}_1=(2,3)$, then $\boldsymbol{Q}_1$ is on the tangent line of $\gamma(t_0)$. The pedal curve of $\gamma(t)$
is $$Pe(\gamma)(t)=\left(\frac{2t^3+6t-2}{4t^2-1},\frac{13t^2-4t}{4t^2-1}\right).$$
In this case, $Pe(\gamma)(t_0)=(2,3)$, it is consistent with $\boldsymbol{Q}_1$. See the green curve in Figure \ref{f_5}.

Let $\boldsymbol{Q}_2=(3,2)$, then $\boldsymbol{Q}_2$ is on the normal line of $\gamma(t_0)$. The contrapedal curve of $\gamma(t)$
is $$CPe(\gamma)(t)=\left(\frac{2t^3+12t^2-5t}{4t^2-1},\frac{4t^4-2t^2+6t-2}{4t^2-1}\right).$$
In this case, $CPe(\gamma)(t_0)=(3,2)$, it is consistent with $\boldsymbol{Q}_2$. See the orange dashed curve in Figure \ref{f_5}.
\begin{figure}[ h ]
\centering
\includegraphics [width =8.5cm]{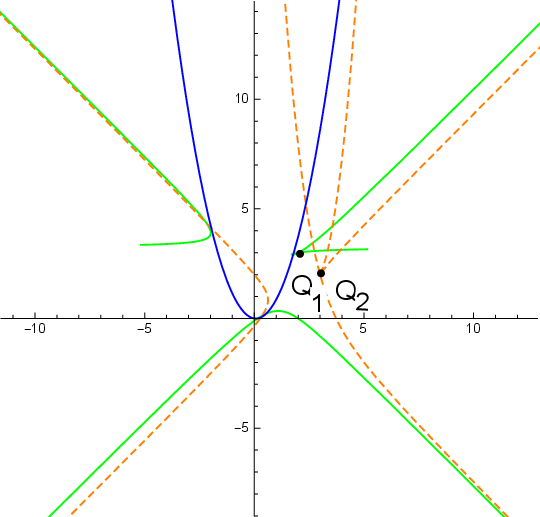}
\caption{The mixed-type curve (blue) and its pedal curve and contrapedal curve.}\label{f_5}
\end{figure}
\end{example}
Consider pedal curves, contrapedal curves and evolute of a mixed-type curve and we can have the theorem as follows.
\begin{theorem}
Let $\gamma:I\rightarrow\mathbb{R}_1^2$ be a regular mixed-type curve, $\boldsymbol{Q}$ be a point in $\mathbb{R}_1^2$, $Pe(\gamma):I\rightarrow\mathbb{R}_1^2$ be the pedal curve of $\gamma$, $CPe(\gamma):I\rightarrow\mathbb{R}_1^2$ be the contrapedal curve of $\gamma$ and $Ev(\gamma):I\rightarrow\mathbb{R}_1^2$ be the evolute of $\gamma$. If $\boldsymbol{Q}$ is consistent with the lightlike point or $\boldsymbol{Q}$ is on the tangent line of the lightlike point, then
$$Pe(Ev(\gamma))(t)=CPe(\gamma)(t).$$
\end{theorem}

\begin{proof}
Since the evolute of a regular mixed-type curve without inflections is given by
$$Ev(\gamma)(t)=\gamma(t)-\frac{2\alpha(t)\beta(t)}{\dot\alpha(t)\beta(t)-\alpha(t)\dot\beta(t)}(\alpha(t)\mathbb{L}^+-\beta(t)\mathbb{L}^-),$$
we have
\pagebreak
\begin{align}\nonumber
&\dot Ev(\gamma)(t)\\ \nonumber
=&\alpha(t)\mathbb{L}^++\beta(t)\mathbb{L}^-\\ \nonumber
&-\frac{2(\dot\alpha(t)\beta(t)+\alpha(t)\dot\beta(t))\dot\alpha(t)\beta(t)-\alpha(t)\dot\beta(t)
-2\alpha(t)\beta(t)\ddot\alpha(t)\beta(t)-\alpha(t)\ddot\beta(t)}
{(\dot\alpha(t)\beta(t)-\alpha(t)\dot\beta(t))^2}\\ \nonumber
&~~~~(\alpha(t)\mathbb{L}^+-\beta(t)\mathbb{L}^-)\\ \nonumber
&-\frac{2\alpha(t)\beta(t)}{\dot\alpha(t)\beta(t)-\alpha(t)\dot\beta(t)}(\dot\alpha(t)\mathbb{L}^+-\dot\beta(t)\mathbb{L}^-)\\ \nonumber
=&\alpha(t)(1-\frac{2\dot\alpha^2(t)\beta^2(t)-2\alpha^2(t)\dot\beta(t)-2\alpha(t)\beta(t)(\ddot\alpha(t)\beta(t)-\alpha(t)\ddot\beta(t))}
{(\dot\alpha(t)\beta(t)-\alpha(t)\dot\beta(t))^2}\\ \nonumber
&-\frac{2\dot\alpha(t)\beta(t)}{\dot\alpha(t)\beta(t)-\alpha(t)\dot\beta(t)})\mathbb{L}^+\\ \nonumber
&+\beta(t)(1+\frac{2\dot\alpha^2(t)\beta^2(t)-2\alpha^2(t)\dot\beta(t)-2\alpha(t)\beta(t)(\ddot\alpha(t)\beta(t)-\alpha(t)\ddot\beta(t))}
{(\dot\alpha(t)\beta(t)-\alpha(t)\dot\beta(t))^2}\\ \nonumber
&+\frac{2\dot\alpha(t)\beta(t)}{\dot\alpha(t)\beta(t)-\alpha(t)\dot\beta(t)})\mathbb{L}^-\\ \nonumber
=&\alpha(t)\frac{-3\dot\alpha^2(t)\beta^2(t)+3\alpha(t)\dot\beta^2(t)+2\alpha(t)\beta(t)(\ddot\alpha(t)\beta(t)-\alpha(t)\ddot\beta(t))}
{(\dot\alpha(t)\beta(t)-\alpha(t)\dot\beta(t))^2}\mathbb{L}^+\\ \nonumber
&-\beta(t)\frac{-3\dot\alpha^2(t)\beta^2(t)+3\alpha(t)\dot\beta^2(t)+2\alpha(t)\beta(t)(\ddot\alpha(t)\beta(t)-\alpha(t)\ddot\beta(t))}
{(\dot\alpha(t)\beta(t)-\alpha(t)\dot\beta(t))^2}\mathbb{L}^-.
\end{align}

Writing
$$\alpha_{Ev}(t)=\alpha(t)\frac{-3\dot\alpha^2(t)\beta^2(t)+3\alpha(t)\dot\beta^2(t)+2\alpha(t)\beta(t)(\ddot\alpha(t)\beta(t)-\alpha(t)\ddot\beta(t))}
{(\dot\alpha(t)\beta(t)-\alpha(t)\dot\beta(t))^2}$$
and
$$\beta_{Ev}(t)=-\beta(t)\frac{-3\dot\alpha^2(t)\beta^2(t)+3\alpha(t)\dot\beta^2(t)+2\alpha(t)\beta(t)(\ddot\alpha(t)\beta(t)-\alpha(t)\ddot\beta(t))}
{(\dot\alpha(t)\beta(t)-\alpha(t)\dot\beta(t))^2}.$$
Then,
\begin{align}\nonumber
&Pe(Ev(\gamma))(t)\\ \nonumber
=&Ev(\gamma)(t)-\frac{\langle\boldsymbol{Q}-Ev(\gamma)(t),\alpha_{Ev}(t)\rangle\mathbb{L}^++\beta_{Ev}(t)\mathbb{L}^-}
{4\alpha_{Ev}(t)\beta_{Ev}(t)}(\alpha_{Ev}(t)\mathbb{L}^++\beta_{Ev}(t)\mathbb{L}^-)\\ \nonumber
=&\gamma(t)+\frac{\langle\boldsymbol{Q}-\gamma(t),\alpha(t){\mathbb{L}^+}-\beta(t){\mathbb{L}^-}\rangle}
{4\alpha(t)\beta(t)}(\alpha(t){\mathbb{L}^+}-\beta(t){\mathbb{L}^-})\\ \nonumber
=&CPe(\gamma)(t).
\end{align}
\end{proof}
According to the above theorem, the relationship among the pedal curve, the contrapedal curve and the evolute of a mixed-type curve in $\mathbb{R}_1^2$ is given.

\bmhead{Acknowledgments}

The work was supported by Natural Science Foundation of Shandong Province (No. ZR2023QA046).

\section*{Declarations}

\begin{itemize}
\item Conflict of interest: The authors declare that they have no conflicts of interest.
\end{itemize}

\bigskip
%
%
%
%

\begin{appendices}




\end{appendices}



\end{document}